\documentclass[preprint]{elsarticle}
\usepackage{natbib}
\usepackage{amsthm}
\usepackage{amsmath}
\usepackage{graphicx}
\usepackage{xcolor}
\usepackage{caption}
\usepackage{subcaption}
\usepackage{tikz}
\usepackage{amssymb}
\usepackage{xurl}
\graphicspath{{/images/}}

\newtheorem{theorem}{Theorem}

\usepackage{todonotes}

\begin{document}
	
	\begin{frontmatter}	
		\title{Strategic Revenue Management of Preemptive versus Non-Preemptive Queues }
		\author{Jonathan Chamberlain\corref{cor1}}
		\ead{jdchambo@bu.edu}
		\author{David Starobinski}
		\ead{staro@bu.edu}
		\address{Department of Electrical and Computer Engineering, Boston University, 8 St Mary's Street, Boston MA 02215}
		
		\cortext[cor1]{Corresponding author}
		
		\begin{abstract}
		Consider a two-class unobservable priority queue, with Poisson arrivals, generally distributed service, and strategic customers. Customers are charged a fee when joining the premium class. We analyze the maximum revenue achievable under the non-preemptive (NP) and preemptive-resume (PR) policies, and show that a provider is always better off implementing the PR policy. Further, the maximum revenue under PR is sometimes achieved when only a fraction of the customers join the premium class.
		\end{abstract}
		
		\begin{keyword}
			Game Theory, Queuing Theory, Preemption, Pricing.
		\end{keyword}
		
	\end{frontmatter}
	
\section{Introduction}

Priority scheduling is utilized in many contexts in order to triage service to those who require more urgent service, or to enable a separate revenue stream by customers paying an additional fee to upgrade to higher priority service. 
In this work, we consider strategic behavior under an $M|G|1$ queue with two classes, where customers have the option to pay a fee to upgrade to the higher priority class. In particular, our focus is on the provider's revenue management, the stability of the resulting equilibrium, and the resulting social welfare. Under this model, we show that under a non-preemptive (NP) policy, the provider always has incentive to charge a fee such that all customers purchase a priority upgrade in order to maximize revenue. In contrast, under preemptive-resume (PR), for sufficiently high variance in service times and low traffic loads, revenues are maximized when the upgrade fee is at a level when only some customers are willing to purchase an upgrade. Otherwise, revenues are maximized when charging a fee where all customers will purchase an upgrade, as in NP.

In any event, the revenue-maximizing fee yields a stable equilibrium, and thus the provider is guaranteed to receive the maximum revenue by setting the fee accordingly. This holds regardless of the preemption policy, variance in service times, or traffic load. Furthermore, for any given traffic load and service time variance the provider will always have incentive to implement the PR policy over the NP one. 

\subsection{Related Work}

Analyses of performance of non-preemptive priority queues has been done for contexts including traffic routing over networks \cite{Walraevens2000}, hospital management \cite{Pekoz2002}, and smart grids \cite{Sadeghi2012}. Under these works however, customers/objects are tagged with class membership based upon predetermined characteristics and cannot purchase the ability to upgrade. Strategic behavior under preemptive-resume models are considered in \cite[pp 83-85]{ToQueueorNottoQueue}, \cite{Shi2019}, \cite{Gurvich2019}. However, these models explicitly or implicitly assume an $M|M|1$ queue to be in effect. Further, while two \cite{Shi2019} or an arbitrary number of classes \cite{Gurvich2019} are considered in the latter two works, there is no option to upgrade between classes, and thus the problem being considered is not identical to the one analyzed here. In \cite[pp 83-85]{ToQueueorNottoQueue}, customers are able to pay an upgrade fee. However, as it too is based on an $M|M|1$ queue, it is asserted that mixed equilibria states will never be stable. As determined in the analysis of the $M|G|1$ version of the model in \cite{chamberlain2020social}, this does not hold in general, which has implications for a revenue maximizing provider as noted below. 

\section{Game Model}
Our model is based on an unobservable $M|G|1$ priority queue, similar to that analyzed in~\cite{chamberlain2020social}. To keep the paper self-contained, we summarize in this section key results of~\cite{chamberlain2020social}, which serve as the basis of our main findings presented in Section~3.
In this queue, customers are by default assigned to the \emph{ordinary} class, but on entry may pay an optional fee $C$ to join the \emph{premium} class. We assume that paying $C$ only impacts a customer's priority class, and not the actual service.  We use the standard notation to denote the mean arrival rate ($\lambda$), mean service rate ($\mu$), and traffic load ($\rho = \lambda/\mu$). In addition, $C$ is used to indicate the fee to join the premium class, $\phi$ denotes the fraction of customers who have joined the premium class, and $K$ is a variance parameter defined such that the second moment of service is equal to $K/\mu^2$. 

We consider customer and provider behavior under the non-preemptive and preemptive-resume policies. The possible equilibria states depend on the customers' decisions, which in turn depend on the time spent waiting for service and the fee $C$. Letting $E[W_p]$ be the time spent waiting for service as a member of the premium class, and $E[W_o]$ be the time spent waiting for service as a member of the ordinary class, there are three possible equilibria states:

\begin{enumerate}
    \item \emph{All join} (i.e. $\phi = 1$): $E[W_p] + C < E[W_o]$ for all $\phi$.
    \item \emph{None join} (i.e. $\phi = 0$): $E[W_p] + C > E[W_o]$ for all $\phi$.
    \item \emph{Some join} (i.e. $E[W_p] + C = E[W_o]$ for some $\phi \in (0,1)$): Customers are indifferent, and join the premium class with probability $\phi$.
\end{enumerate}



\subsection{Equilibria under the NP policy}

The cost function associated with the NP policy is given as follows:
\begin{equation}
    \mathcal{C}_{NP}(\phi) = \frac{K\rho^2}{2\mu(1-\rho)(1-\phi\rho)}.
    \label{eq: C NP}
\end{equation}

Evaluating this function, we find that regardless of the values of the parameters $K$ and $\rho$, this is always a monotone increasing function in $\phi$. Thus it has the following equilibrium structure. If $C < \mathcal{C}_{NP}(0)$, \emph{all join} is the unique equilibrium. If $C > \mathcal{C}_{NP}(1)$, \emph{none join} is the unique equilibrium. If $\mathcal{C}_{NP}(0) < C < \mathcal{C}_{NP}(1)$, \emph{all join}, \emph{none join}, and \emph{some join} equilibrium $\phi = 1/\rho - (K\rho)/(2\mu C(1-\rho)$ are all possible equilibria. The pure \emph{all join} and \emph{none join} equilibria are always stable. The mixed \emph{some join} equilibrium is never stable.


\subsection{Equilibria under the PR policy}
The cost function for the PR policy is derived as follows:
\begin{equation}
    \mathcal{C}_{PR}(\phi) = \frac{K\rho + (2-K)\phi\rho(1-\rho)}{2\mu(1-\rho)(1-\phi\rho)}.
    \label{eq: C PR}
\end{equation}

Unlike $\mathcal{C}_{NP}(\phi)$, this function's behavior depends on the values of $K$ and $\rho$, although it is always monotone or constant:
\begin{enumerate}
    \item If $K > 2$ and $\rho < (K-2)/(2K-2)$, $\mathcal{C}_{PR}(\phi)$ is monotone decreasing.
    \item Else, if $K > 2$ and $\rho = (K-2)/(2K-2)$, $\mathcal{C}_{PR}(\phi)$ is constant valued.
    \item Otherwise, $\mathcal{C}_{PR}(\phi)$ is monotone increasing. 
\end{enumerate}

This in turn influences the equilibrium structure under preemptive-resume. In particular, there are circumstances in which the \emph{some join} equilibrium will be stable. If the \emph{some join} equilibrium exists, it is the solution to $\mathcal{C}_{PR}(\phi) = C$, and we denote this by 
\begin{equation*}
    \phi_e = \frac{2\mu C(1-\rho)K\rho}{\rho(1-\rho)(2\mu C + 2 - K)}. 
\end{equation*}
If $C < \min\mathcal{C}(\phi)$, \emph{all join} is the unique equilibrium. If $C > \max\mathcal{C}(\phi)$, \emph{none join} is the unique equilibrium. If $\min \mathcal{C}(\phi) < C < \max \mathcal{C}(\phi)$, and $\mathcal{C}(\phi)$ is monotone decreasing, the \emph{some join} equilibrium $\phi = \phi_e$ is the unique equilibrium. If $\min \mathcal{C}(\phi) < C < \max \mathcal{C}(\phi)$, and $\mathcal{C}(\phi)$ is monotone increasing, the equilibrium outcomes are similar to those of the NP policy.


With the equilibria states established, we now turn to analyzing the revenues collected from customers upgrading to the premium class under each preemption policy.

\section{Revenue Management}

As noted, the customers' decision of whether to join the premium class or not depends on the cost of waiting as a member of each class. This decision is influenced by the joining fee $C$, which is controlled by the provider. As the provider is rational, they are interested in maximizing their revenues, and thus will set $C$ appropriate. However, as seen in the previous section, not all possible equilibria are stable. In particular, any time a \emph{some join} state is unstable, a \emph{none join} equilibrium is possible. In such a case, there is a risk the provider would not collect any revenue from customers upgrading to the premium class. 

Thus, we define and analyze a revenue function for each preemption policy in order to determine the equilibrium corresponding to the provider's maximum revenue. We also whether whether that equilibrium is stable. While revenue is defined in terms of the cost and the total number of customers purchasing the upgrade, the latter is not knowable until the completion of service. Thus, we define the revenue $\mathcal{R}(\phi)$ in terms of average revenue per time unit. Given an arrival rate of $\lambda$, and the fraction $\phi$ of customers, we derive the revenue function as
\begin{equation}
    \mathcal{R}(\phi) = \lambda\phi\mathcal{C}(\phi).
    \label{eq: R phi definition}
\end{equation}

We further denote the maximum revenue under each preemption policy as $\mathcal{R}^* = \max_{\phi \in [0,1]} \mathcal{R}(\phi)$.

\subsection{NP Model}

Under the NP policy, the corresponding revenue function is derived from the definition in Equation~\eqref{eq: R phi definition} and $\mathcal{C}_{NP}(\phi)$ as follows:
\begin{equation}
    \mathcal{R}_{NP}(\phi) = \frac{K\rho^3\phi}{2(1-\rho)(1-\phi\rho)}.
    \label{eq: R NP}
\end{equation}

Evaluating the derivative of $\mathcal{R}_{NP}(\phi)$, we determine that for all $K$ and $\rho$, the function is monotone increasing. As a result, the maximum revenue is obtained when $\phi = 1$. The resulting maximum revenue is equal to
\begin{equation}
    \mathcal{R}_{NP}^* = \mathcal{R}_{NP}(1) = \frac{K\rho^3}{2(1-\rho)^2}.
    \label{eq: R NP max}
\end{equation}

Thus, under the non-preemptive policy, the variance in service and the traffic load will influence the revenue collected, but the provider will always have incentive to set the fee $C$ to be equal to $\mathcal{C}_{NP}(1)$. As shown below, this is not necessarily the case under the preemptive-resume policy. 

\subsection{PR Model}

Under the PR policy, the corresponding revenue function is derived from Equations~\eqref{eq: R phi definition} and~\eqref{eq: C PR} as
\begin{equation}
    \mathcal{R}_{PR}(\phi) = \frac{K\phi\rho^2 + (2-K)\phi^2\rho^2(1-\rho)}{2(1-\rho)(1-\phi\rho)}
    \label{eq: R PR}
\end{equation}

Unlike under NP, this revenue function's behavior does depend on the values of $K$ and $\rho$. To show this, we compute the derivative with respect to $\phi$:
\begin{equation}
    \mathcal{R}_{PR}'(\phi) = \frac{\rho^2\Big(2\phi(1-\rho)(2-\phi\rho)+K(1-\phi(1-\rho)(2-\phi\rho)\Big)}{2(1-\rho)(1-\phi\rho)^2}.
\end{equation}

Evaluating the expression, we determine that the sign of the derivative (and thus the increasing/decreasing behavior of $\mathcal{R}_{PR}(\phi)$) depends solely on the sign of the expression 
\begin{equation*}
    2\phi(1-\rho)(2-\phi\rho)+K(1-\phi(1-\rho)(2-\phi\rho).
\end{equation*} Given the restrictions on the parameters, $K \in [1,\infty)$ and $\rho \in (0,1)$, we solve for where the expression is positive to determine where $\mathcal{R}_{PR}(\phi)$ is increasing. In doing so, we determine the following: If $1 \leq K \leq 4$, $\mathcal{R}_{PR}(\phi)$ is monotone increasing regardless of the value of $\rho$. If $K > 4$, $\mathcal{R}_{PR}(\phi)$ is monotone increasing if $\rho \in [3/2 - (1/2)\sqrt{(5K-2)/(K-2)},1)$. If $K > 4$, $\mathcal{R}_{PR}(\phi)$ is unimodal with a unique maximum if $\rho \in (0,3/2 - (1/2)\sqrt{(5K-2)/(K-2)})$.

In the first two cases, the monotone increasing behavior of $\mathcal{R}_{PR}(\phi)$ results in the maximum revenue being achieved when $\phi = 1$. Thus, $\mathcal{R}_{PR}^* = \mathcal{R}_{PR}(1)$, which equals
\begin{equation}
     \frac{K\rho^2 + (2-K)\rho^2(1-\rho)}{2(1-\rho)^2}.
     \label{eq: R PR max 1}
\end{equation}

However, when $K > 4$, and $\rho < 3/2 - (1/2)\sqrt{(5K-2)/(K-2)}$, $\mathcal{R}_{PR}(\phi)$ is unimodal with a unique maximum. As a result, the revenue is maximized at a $\emph{some join}$ equilibrium. The value of $\phi$ which maximizes $\mathcal{R}_{PR}(\phi)$ is
\begin{equation*}
    \phi^{max} = \frac{1}{\rho}\Bigg(1 - \sqrt{\frac{K-2-2\rho(K-1)}{(K-2)(1-\rho)}}\Bigg).
\end{equation*} 
The corresponding maximum revenue is equal to 
\begin{equation}
    \mathcal{R}_{PR}^* = \frac{2(K-2)-\rho(3K-4)}{2(1-\rho)}-(K-2)\sqrt{\frac{K-2-2\rho(K-1)}{(K-2)(1-\rho)}}.
    \label{eq: R PR max 2}
\end{equation}

\subsection{Equilibrium Stability and Guarantee of Maximum Revenue}

Having determined the values of the maximum possible revenue, we face the question of whether the provider is actually guaranteed to receive such revenue if the fee $C$ is set accordingly, due to the existence of potentially unstable equilibria. However, we assert that this is not the case here.
\begin{theorem}
    The provider is always guaranteed to receive their maximum possible revenue if $C$ is set accordingly, regardless of the preemption policy, variance in service times, or traffic load. 
\end{theorem}

\begin{proof}
    For the NP policy, this is straightforward, since per Equation \eqref{eq: R NP max}, the maximum revenue is obtained when $\phi = 1$. This corresponds to the \emph{all join} equilibrium, which is always stable, and thus the revenue is guaranteed.
    
    In the PR policy, this holds for the same reason so long as $\mathcal{R}_{PR}(\phi)$ is monotone increasing. Thus, assume that $\mathcal{R}_{PR}(\phi)$ is unimodal with a unique maximum. In this case, $\phi^{max}$ is a mixed equilibrium. Mixed equilibria are stable so long as the corresponding cost function $\mathcal{C}_{PR}(\phi)$ is itself monotone decreasing. 
    
    Thus, if $\mathcal{C}_{PR}(\phi)$ is shown to be monotone decreasing under these circumstances, then $\phi^{max}$ is ESS stable and thus the corresponding revenue is guaranteed. $\mathcal{C}_{PR}(\phi)$ is monotone decreasing when $K > 2$ and $\rho < (K-2)/(2K-2)$. $\mathcal{R}_{PR}(\phi)$ is unimodal if $K > 4$ and $\rho < 3/2 - (1/2)\sqrt{(5K-2)/(K-2)}$.
    
    Thus, clearly the condition on the service variance parameter $K$ is satisfied, and thus we must simply show that for $K > 4$, the inequality 
    \begin{equation*}
        \frac{3}{2} - \frac{1}{2}\sqrt{\frac{5K-2}{K-2}} < \frac{K-2}{2K-2}
    \end{equation*}
    holds. Solving for $K$, we determine that the inequality holds if $K < 0$ or $K > 2$, thus it certainly holds when $K > 4$. Therefore the traffic load resulting in a unimodal $\mathcal{R}_{PR}(\phi)$ is in the range which results in a monotone decreasing $\mathcal{C}_{PR}(\phi)$. Therefore, the resulting revenue maximizing equilibrium $\phi^{max}$ is stable and therefore the maximum revenue is guaranteed.
\end{proof}

With the maximum revenues guaranteed to be obtained if $C$ is set accordingly by the provider, we next consider which preemption policy will result in the greatest revenues from customer upgrades.


\subsection{Comparison of Maximum Revenues}

As the provider is rational and selects the preemption policy in effect, they are incentivized to select the policy which corresponds to the greater maximum revenue for a given traffic load $\rho$ and variance in service as denoted by $K$. We claim that regardless of the traffic load or variance in service, a rational revenue maximizing provider will always implement a preemptive-resume policy:

\begin{theorem}
    Consider a two class $M|G|1$ queuing model, where customers pay a fee $C$ if they wish to upgrade to the premium class. The provider is always better off implementing the preemptive-resume policy, as the maximum revenue under the non-preemptive policy is always lower.
    \label{thm: comp revs}
\end{theorem}

\begin{proof}
    Let $\rho$ and $K$ be arbitrary but fixed. There are two cases to consider based on the behavior of the corresponding $\mathcal{R}_{PR}(\phi)$.
    
    Assume that $\mathcal{R}_{PR}(\phi)$ is monotone increasing, thus $\mathcal{R}_{PR}^*$ is defined as in Equation \eqref{eq: R PR max 1}. Comparing this to $\mathcal{R}_{NP}^*$ as defined in Equation \eqref{eq: R NP max}, we claim the following holds if $\mathcal{R}_{PR}(\phi)$ is monotone increasing:
    \begin{equation*}
        \frac{K\rho^2 + (2-K)\rho^2(1-\rho)}{2(1-\rho)^2} > \frac{K\rho^3}{2(1-\rho)^2}.
    \end{equation*}
    This reduces to determining whether $2(1-\rho)\rho^2 > 0$, which is true for all $\rho \in (0,1)$, and is independent of $K$. Thus, $\mathcal{R}_{PR}^* > \mathcal{R}_{NP}^*$ in the case where $\mathcal{R}_{PR}(\phi)$ is monotone increasing. In fact, this result shows that $\mathcal{R}_{PR}(1) > \mathcal{R}_{PR}^*$ in general. This implies that as expected, $\mathcal{R}_{PR}^* > \mathcal{R}_{NP}^*$ as well when $\mathcal{R}_{PR}(\phi)$ is unimodal. 
    
\end{proof}

Therefore, we find that regardless of the variance in service time distribution, and regardless of the traffic load, the provider is always best off implementing the preemptive-resume policy in order to maximize their revenues. 

\subsection{Impact on Social Welfare}

We now turn to a brief analysis of how the maximum revenue situation impacts the social welfare of the system. Under this model, the social welfare is defined in terms of the average wait time across all customers, as all other costs are either fixed or are a transfer of payment between players \cite{chamberlain2020social}:
\begin{equation}
    \mathcal{S}(\phi) = \phi E[W_p] + (1-\phi) E[W_o]
    \label{eq: social welfare definition}
\end{equation}

Under the non-preemptive policy, we determine that the social welfare is constant valued, a result we expect as wait times in an $M|G|1$-NP queue are known to be constant with respect to the reordering of customers:
\begin{equation}
    \mathcal{S}_{NP}(\phi) = \frac{K\rho}{2\mu(1-\rho)}.
\end{equation}
Thus, under NP, the provider does not impact the overall welfare by behaving in a revenue maximizing fashion. However, under the preemptive-resume policy, social welfare is not constant as preemption behaviors impact the overall wait times due to the differences between the service time of the preempting customer, and the residual service of the customer being preempted. The social welfare function under PR is derived from the definition as follows

\begin{equation}
    \mathcal{S}_{PR}(\phi) = \frac{\rho\Big(K - 2\phi\rho + (2-K)\phi(1-\phi(1-\rho)\Big)}{2\mu(1-\rho)(1-\phi\rho)}
\end{equation}

Here, the value of $K$ determines which states are socially optimal: If $K < 2$, the socially optimal states are the pure equilibria states $\phi \in \{0,1\}$. If $K = 2$, all equilibria states are socially optimal. If $K > 2$, the socially optimal state is the \emph{some join} equilibrium
    \begin{equation*}
        \phi^* = \frac{1-\sqrt{1-\rho}}{\rho}.
    \end{equation*} 


From this, we conclude that a revenue maximizing provider also acts in the socially optimal manner whenever $K \leq 2$. If $K > 2$, this is not the case. In fact, for $2 < K \leq 4$, we find that a revenue maximizing provider is in the worst case social outcome, as the \emph{all join} equilibrium leads to the longest expected wait times. When $K > 4$, under sufficiently low traffic loads the worst case outcome is avoided. However, the revenue maximizing equilibrium will never be socially optimal under these circumstances.

\section{Conclusions}

In this work, we analyzed an $M|G|1$ queue with two priority classes where customers are strategic and have the ability to upgrade to a premium class. We compared the maximum provider's revenue under the non-preemptive and preemptive-resume policies, and showed that a rational provider is always best off implementing the preemptive-resume policy. Further, under sufficiently high variance in service and low traffic loads, there exist scenarios where the revenue is maximized by only a fraction of the customers joining the premium class. 

\section*{Acknowledgements}
This work was  partially supported by NSF grants 1717858 and 1908087.
	
	\bibliographystyle{elsarticle-num}
	\bibliography{bibfile}
	
\end{document}